\documentclass[letterpaper, 12pt, conference]{ieeeconf}  

\IEEEoverridecommandlockouts                              
\usepackage{graphicx} 
\usepackage{mathptmx} 
\usepackage{times} 
\usepackage{amsmath} 
\usepackage{amssymb}  
\newtheorem{lem}{Lemma}
\newtheorem{thm}{Theorem}

\def\ad{{\rm ad}}           
\def\<{\leqslant}           
\def\>{\geqslant}           

\def\d{\partial}

\def\Re{{\rm Re}}   
\def\Im{{\rm Im}}   

\def\cH{{\cal H}}   
\def\mA{{\mathbb A}}    
\def\mR{{\mathbb R}}    
\def\mC{{\mathbb C}}    

\def\Tr{{\rm Tr}}       
\def\rT{{\rm T}}        

\def\bE{{\bf E}}    

\def\[[[{[\![\![}
\def\]]]{]\!]\!]}

\def\bra{{\langle}}
\def\ket{{\rangle}}

\def\Bra{\big\langle}
\def\Ket{\big\rangle}

\def\re{{\rm e}}        
\def\rd{{\rm d}}        

\def\cL{{\mathcal L}}

\def\x{\times}
\def\ox{\otimes}
\def\od{\odot}

\def\om{{\ominus}}

\def\cF{{\cal F}}

\def\cK{{\cal K}}

\def\cE{{\cal E}}

\def\mS{{\mathbb S}}

\def\ups{\upsilon}
\def\Ups{\Upsilon}

\def\diag{\mathop{\rm diag}}    

\title{\Large \bf
Risk-sensitive Dissipativity of Linear Quantum Stochastic Systems under Lur'e Type
Perturbations of Hamiltonians
}

\author{Igor G. Vladimirov, \qquad Ian R. Petersen
\thanks{This work is supported by the Australian Research Council. The authors are
with the School of Engineering and Information Technology, University of New South Wales at the Australian Defence Force Academy, Canberra, ACT 2600, Australia. E-mail: {\tt igor.g.vladimirov@gmail.com, i.r.petersen@gmail.com}.
}
}

\onecolumn
\begin{document}
\maketitle
\thispagestyle{empty}




\begin{abstract}
This paper is concerned with
a stochastic dissipativity theory using quadratic-exponential storage functions for open quantum systems with canonically commuting dynamic variables
governed by quantum stochastic differential equations. The system is linearly coupled to external boson fields and has a quadratic Hamiltonian which is perturbed by nonquadratic functions of linear combinations of system variables.  Such perturbations are similar to those in the classical Lur'e systems  and make the quantum dynamics nonlinear. We study their effect
on the quantum expectation of the exponential of a positive definite  quadratic form of the system variables.
This allows conditions to be established
for the risk-sensitive stochastic storage function of the quantum system  to remain bounded, thus securing boundedness for the moments of system variables of arbitrary order.
These results employ a noncommutative analogue of the Doleans-Dade exponential
and a multivariate partial differential version of the Gronwall-Bellman lemma.

\end{abstract}


\section{Introduction}\label{sec:intro}

Quantum systems, originated in quantum physics and its applications, including, for example,
the  interaction of coherent light with matter at an atomic level studied in quantum optics \cite{GZ_2004,WM_2008}, can be described briefly as noncommutative stochastic systems and are treated by quantum probability tools.
 Their dynamic variables are represented in terms of an algebra of (generally, noncommuting) operators on a Hilbert space, with self-adjoint operators (usually referred to as observables) corresponding to real-valued physical quantities. The role of a probability measure  is played, though in a noncommutative fashion, by a density operator \cite{
 M_1998} on the underlying Hilbert space which specifies the quantum  state of  the system.
An efficient language to describe open quantum systems, interacting with the environment, is provided by quantum stochastic differential equations (QSDEs) \cite{HP_1984,P_1992} (see also a review paper \cite{H_1991} and references therein) which govern the system variables in the Heisenberg picture of quantum dynamics. The QSDEs are driven by a quantum noise from the surrounding memoryless heat bath, represented by a boson Fock space \cite{P_1992}, and employ a system-bath interaction model in combination with the Hamiltonian which would specify the ``internal'' dynamics of the system in isolation from the external fields. This unified Markovian approach to open quantum systems and their interconnections is widely used in quantum control  \cite{B_1983,DDJW_2006,EB_2005,J_2005,JG_2010,JNP_2008} (see \cite{ZJ_2012} for a more complete bibliography)  which is aimed at synthesizing such systems and  enhancing their performance (including stability, optimality, robustness) for practical applications, such as quantum metrology and  quantum-optical communication, to mention a few \cite{NC_2000}.
An important class of open quantum systems is formed by those where  the system variables satisfy canonical commutation relations (CCRs) (similarly to the quantum-mechanical position and momentum operators  \cite{M_1998}), the system Hamiltonian is quadratic and the system-bath coupling operators are linear with respect to the variables. Such systems are governed by linear QSDEs and are dynamically equivalent to open quantum harmonic oscillators which are basic models in linear quantum stochastic control \cite{EB_2005,NJP_2009,P_2010,VP_2011a}. Despite their theoretic  convenience (including closed-form solutions and  preservation of Gaussian nature of the system density operator \cite{JK_1998,P_2010}), the linear quantum dynamics remain a simplified model, which may be distorted by an unmodelled energetics of the real system and its environment. The resulting perturbed  behaviour can be ``localized'' in terms of bounds for the second moments of system variables, which play the role of storage functions to quantify robust  mean square  stability \cite{PUJ_2012} of the system in the quantum stochastic version  \cite{JG_2010} of the classical dissipativity theory \cite{W_1972}.
In the present paper, we develop this line of research further by extending risk-sensitive storage functions \cite{MP_2009}  from classical stochastic systems to the quantum setting. The Hamiltonian of the quantum systems being considered has a nominal quadratic part which is perturbed by nonquadratic functions of linear combinations of the system variables.  Such perturbations are similar to those in the classical Lur'e systems  \cite{L_1951} and lead to nonlinear QSDEs  which neither lend themselves to closed-form solution nor maintain Gaussian quantum statistics. We study the dynamics  
of the quantum expectation of the exponential of a positive definite quadratic form of the system variables. For this purpose, we use a noncommutative analogue of the Doleans-Dade exponential \cite{D_1970} based on a commutator approach to parameter differentiation of exponential operators \cite{L_1968,M_1954,S_1964,W_1967} to take advantage of the CCRs between the system variables. In combination with a multivariate partial differential version of the Gronwall-Bellman lemma, this allows conditions to be established for the perturbations of the Hamiltonian under which the risk-sensitive storage function of the quantum system  remains bounded, thus securing boundedness of moments of the system variables of arbitrary order. The methods, developed  in this paper, differ from those in \cite{DDJW_2006,J_2005,YB_2009} and can be of interest for the design of coherent (measurement-free)   risk-sensitive quantum control schemes.

\section{Underlying quantum stochastic systems}\label{sec:class}

We consider an open quantum system with $n$ dynamic variables $X_1(t),\ldots, X_n(t)$, which evolve in time $t\> 0$ and are assembled into a vector $X(t):= (X_j(t))_{1\< j\< n}$ (vectors are  organised as columns unless indicated otherwise).   The system variables $X_1(0),\ldots, X_n(0)$ at the initial moment of time  are self-adjoint operators on a complex separable Hilbert space $\cH$ (such as the quantum-mechanical position and momentum operators) which are assumed to satisfy CCRs
\begin{equation}
\label{Theta}
    [X, X^{\rT}]
    :=
    (
        [X_j,X_k]
    )_{1\< j,k\< n}
    =
    XX^{\rT}- (XX^{\rT})^{\rT}
    =
    i
    \Theta.
\end{equation}
Here, $[A, B]:= AB-BA$ is the commutator of operators, and the transpose $(\cdot)^{\rT}$ applies to matrices with operator-valued entries as if the latter were scalars.  Also, $i:= \sqrt{-1}$ is the imaginary unit, and $\Theta:= (\theta_{jk})_{1\<j, k\< n}$ is a real antisymmetric matrix of order $n$ (the space of such matrices is denoted by $\mA_n$).
By the standard convention, a linear  operator $\xi$ on $\cH$ is lifted to its ampliation $\xi\ox I_{\cF}$ on the tensor  product space $\cH \ox \cF$, where $I_{\cF}$ denotes the identity operator on a boson Fock space $\cF$  \cite{P_1992}. The latter provides a domain for the action of an
$m$-dimensional  quantum Wiener process $W(t):= (W_k(t))_{1\< k \< m}$ adapted to the filtration associated with the continuous tensor product structure of the Fock space. 
The entries $W_1(t), \ldots, W_m(t)$ of the vector $W(t)$ are self-adjoint operators on $\cF$, which are associated with the annihilation and creation operator  processes of external boson fields.
Omitting the time argument, suppose the quantum Ito table of $W$ is given by
\begin{equation}
\label{Omega}
    \rd W\rd W^{\rT}
    :=
    (\rd W_j \rd W_k)_{1\< j,k\< m}
    =
    \Omega \rd t,
\end{equation}
where $\Omega:= (\omega_{jk})_{1\< j,k\< m}$ is a constant complex positive semi-definite Hermitian matrix of order $m$.
Its  entrywise real part
\begin{equation}
\label{V}
    V
    :=
    \Re \Omega
    =
    (\Omega + \Omega^{\rT})/2
\end{equation}
is a positive semi-definite symmetric matrix since $\Omega = \Omega^* \succcurlyeq 0$, with $(\cdot)^*:= (\overline{(\cdot)})^{\rT}$ the complex conjugate transpose.
The imaginary part of the quantum Ito matrix  $\Omega$  in (\ref{Omega}) specifies the CCRs between the fields as
\begin{equation}
\label{CCRW}
    [\rd W, \rd W^{\rT}]
      :=
    ([\rd W_j, \rd W_k] )_{1\< j,k\< m}\\
     =  i J\rd t,
    \quad
    J:= 2 \Im \Omega,
\end{equation}
where  $J\in \mA_m$ in view of the Hermiticity of $\Omega$.  The quantum Wiener process $W$ represents a quantum noise which drives a QSDE
\begin{equation}
\label{XQSDEgen}
    \rd X =
    F\rd t + G\rd W,
\end{equation}
governing the Heisenberg dynamics of the vector $X$ of system variables. The $n$-dimensional
drift vector $F(t)$ and the dispersion $(n\x m)$-matrix $G(t)$ of this QSDE are expressed  as
\begin{equation}
\label{FG}
    F
    :=
    i[H,X] +\cL(X),
    \qquad
        G
    :=
    - i[X,h^{\rT}]
\end{equation}
in terms of a system Hamiltonian $H$, the quantum Ito matrix  $\Omega$ of the process $W$ and system-field coupling operators $h_1, \ldots, h_m$ which are assembled into a vector $h:= (h_k)_{1\< k \< m}$. Here, $\cL$ denotes the Lindblad superoperator \cite{
L_1976} whose action on the vector $X$ is given by
\begin{equation}
\label{cL}
    \cL(X)
    :=
    \frac{1}{2}
    \sum_{j,k=1}^m
    \omega_{jk}
    \big(
        h_j[X,h_k] + [h_j,X]h_k
    \big).
\end{equation}
Regardless of a particular form of $H$ and $h$ (which usually are functions of system variables),  the CCR matrix $\Theta$ in (\ref{Theta}) is preserved in time due to the Heisenberg unitary evolution
\begin{equation}
\label{zetat}
    \zeta(t) =
    U(t)^{\dagger} \zeta(0) U(t)
\end{equation}
of observables $\zeta$  on the system-field composite Hilbert space  $\cH\ox \cF$ (in particular, the lifted system variables $X_k(0)\ox I_{\cF}$). Here, $U(t)$ is a unitary operator on $\cH\ox \cF$ with initial condition  $U(0)= I_{\cH \ox \cF}$, and $(\cdot)^{\dagger}$ is the operator adjoint. The quantum expectation of the observable (\ref{zetat}) is defined as
\begin{equation}
\label{bE}
    \bE\zeta(t)
    :=
    \Tr(\rho(0) \zeta(t) ),
\end{equation}
where the initial system-bath density operator $\rho(0)$ is assumed to be the tensor product
\begin{equation}
\label{rho0}
    \rho(0)
    :=
    \varpi(0)
    \ox
    \ups
\end{equation}
of the initial plant state $\varpi(0)$ on $\cH$ and the vacuum state $\ups := |0\ket\bra 0|$ of the external fields associated with the vacuum vector $|0\ket$ in $\cF$, where the Dirac bra-ket notation \cite{M_1998} is used. In the next section, we will specify the energetics of the system being considered.

%
%
%
%
%


\section{Hamiltonians with Lur'e type perturbations}\label{sec:lur'e}

For what follows, the system-field coupling operators $h_1, \ldots, h_m$  in (\ref{FG}) are assumed to be linear  with respect to   the system variables:
\begin{equation}
\label{hM}
    h:= MX,
\end{equation}
where $M\in \mR^{m\x n}$ is a constant matrix. That is, $h_j$ is a linear combination of $X_1, \ldots, X_n$ whose coefficients form the $j$th row of $M$. Then, in view of the CCRs (\ref{Theta}) and the bilinearity of the commutator,  the dispersion matrix $G$ is a constant real matrix:
\begin{equation}
\label{BM}
    G
    =
    -i[X, X^{\rT}]M^{\rT}
    =
    \Theta M^{\rT}
    =: B.
\end{equation}
Also, suppose the system Hamiltonian $H$ in (\ref{FG}) is described by
\begin{equation}
\label{H}
    H
    := H_0 + \sum_{k=1}^s\varphi_k(Y_k),
    \qquad
    Y_k:= c_k^{\rT} X,
\end{equation}
where
\begin{equation}
\label{H0}
    H_0
    :=
    X^{\rT} R X/2
\end{equation}
is a \emph{nominal} Hamiltonian, specified by a real symmetric matrix $R$ of order $n$ (the space of such matrices is denoted by $\mS_n$). Here,
 $\varphi_1, \ldots, \varphi_s: \mR\to \mR$ are continuously differentiable functions which are applied to self-adjoint operators $Y_1, \ldots, Y_s$, assembled into a vector
 \begin{equation}
 \label{YC}
    Y := (Y_k)_{1\< k\< s}=C^{\rT} X,
    \qquad
    C:=
    \begin{bmatrix}
        c_1 & \ldots & c_s
    \end{bmatrix},
 \end{equation}
 where the columns of the $(n\x s)$-matrix $C$ are the vectors $c_1, \ldots, c_s\in \mR^n$ from (\ref{H}).
In combination with the linear system-field coupling (\ref{hM}), the quadratic Hamiltonian $H_0$ would lead to a linear quantum system,  dynamically equivalent to the open quantum harmonic oscillator \cite{EB_2005,GZ_2004} which is a basic model in linear quantum control \cite{JNP_2008,NJP_2009,P_2010}.
Therefore,  the functions $\varphi_1, \ldots, \varphi_s$, which, in general, are  not quadratic,  can be  interpreted as an unmodelled part of the Hamiltonian (\ref{H}) playing the role of a perturbation to the nominal quadratic  Hamiltonian $H_0$. The following lemma specializes the QSDE (\ref{XQSDEgen}) for the system under consideration.

\begin{lem}
\label{lem:dX}
For the open quantum system with the CCRs (\ref{Theta}), linear system-field coupling operators (\ref{hM}) and Hamiltonian (\ref{H}), the vector $X$ of system  variables satisfies the QSDE
\begin{equation}
\label{dX}
    \rd X = F\rd t + B \rd W.
\end{equation}
Here,
\begin{equation}
\label{FAB}
    F = A X + \Theta CZ,
    \qquad
    A := \Theta R + BJM/2,
\end{equation}
where the matrices $J$ and $B$ are defined by (\ref{CCRW}) and (\ref{BM}), and
\begin{equation}
\label{Z}
    Z
    :=
    (Z_k)_{1\< k \< s},
    \qquad
    Z_k := \varphi'_k(Y_k),
\end{equation}
with the derivatives  $\varphi_k'$ of the functions $\varphi_k$ being evaluated at the operators   $Y_k$.
\end{lem}

\begin{proof}
Being a particular form of (\ref{XQSDEgen}),  the QSDE (\ref{dX}) is obtained by substituting (\ref{hM})--(\ref{H}) into (\ref{FG})--(\ref{cL}). 
The linear part $AX$ of the drift vector $F$ in (\ref{FAB}) comes from the quadratic Hamiltonian $H_0$ in (\ref{H0}) and the linear system-field coupling operators (\ref{hM}) as
\begin{equation}
\label{H0cL}
    i[H_0, X] = \Theta R X,
    \qquad
    \cL(X) = BJMX/2,
\end{equation}
whose derivation is well-known in one form or another  \cite{EB_2005}.
 The term $\Theta CZ$ in (\ref{FAB}) originates from the relationship
$    i[\varphi(Y_k), X]
    =
    i\varphi'(Y_k)[Y_k, X]
    =
    \Theta c_k Z_k
$,
which follows from the commutator identities \cite[pp. 38--39]{M_1998}, the notation (\ref{Z}) and the CCR
$
    [Y_k, X] = -[X,Y_k]
    = - [X,X^{\rT}] c_k = -i\Theta c_k
$
between $Y_k$ and $X$ in view of (\ref{Theta}).  \end{proof}

Therefore, the effect of quadratic perturbations described by
\begin{equation}
\label{quadphi}
    \varphi_k(y)
    = \gamma_k y^2/2,
\end{equation}
with constant coefficients $\gamma_1, \ldots, \gamma_s \in \mR$, is equivalent to modifying the nominal quadratic Hamiltonian $H_0$ in (\ref{H}) as $R\mapsto R+C\diag_{1\< k\< s}(\gamma_k) C^{\rT}$, in which case, the drift vector $F$ in (\ref{FAB}) inherits the  linear dependence on the system variables from (\ref{H0cL}). If the functions $\varphi_1, \ldots, \varphi_s$ are not quadratic,  $F$ becomes nonlinear with respect to  $X$. Such perturbations are similar to those in the classical Lur'e systems \cite{L_1951}. In Sections~\ref{sec:RSM}--\ref{sec:super}, we will investigate the influence of the Lur'e type perturbations on the quadratic-exponential moments of the system variables. This study will be based on the more general results of the next section on exponential moments of adapted quantum processes.

\section{Exponential moments of quantum processes}\label{sec:exp_mom}

Let $\xi$ be an adapted  quantum process on the system-field product space $\cH \ox \cF$  satisfying a QSDE
\begin{equation}
\label{dxi}
  \rd \xi
  =
  f\rd t + g^{\rT}\rd W,
\end{equation}
driven by the quantum Wiener process
$W$ with the Ito matrix $\Omega$ from (\ref{Omega}), where the drift
 $f$ and the dispersion vector $g := (g_j)_{1\< j\< m}$ are adapted quantum processes. We assume that $f(t)$, $g_1(t), \ldots, g_m(t)$ are self-adjoint operators on $\cH\ox \cF$ for any $t\> 0$, and so is $\xi(t)$.
 Such a QSDE can be obtained from (\ref{XQSDEgen}), or its specialization (\ref{dX})--(\ref{Z}), if $\xi$ is a function (for example, a polynomial with real coefficients) of the system variables. In this  case, bounds for the moments $\bE (\xi^r)$ of $\xi$, computed for positive integers $r$ over the vacuum state of the external fields in the sense of (\ref{bE}),  (\ref{rho0}),  would guarantee a statistically  ``localized'' behaviour, that is, robust stability,  of the quantum system \cite{JG_2010,PUJ_2012}. However, more subtle information on the system dynamics is provided by the \emph{exponential moment}
 \begin{equation}
\label{Xi}
    \Xi(t)
    :=
    \bE \re^{\xi(t)}
    =
    \Tr (\rho(0)\re^{\xi(t)}).
\end{equation}
Here, $\xi$ may be multiplied by a real, which, for simplicity, is not done  in the present study since the required effect can be achieved by an appropriate scaling of the processes $f$ and $g$ in (\ref{dxi}).  The following Lemma \ref{lem:EE} employs the ideas of \cite{L_1968}, \cite{M_1954}, \cite[Appendix B]{S_1964}, \cite{W_1967} on parameter differentiation of exponential operators to develop a representation for the quantum Ito differential of the exponential
$     \re^{\xi}
    =
    \sum_{r\>0}\xi^r/r!
$;
see also \cite[pp. 480--482]{K_1976}, \cite[pp. 200--206]{RS_1975} 
on the
exponentials of unbounded self-adjoint operators.
A  straightforward computation of $\rd \re^{\xi}$ is complicated by the noncommutativity between $\xi$ and $\rd \xi$.
To formulate the lemma, let $\cE_{\lambda}$ denote a linear superoperator,  associated
with $\xi(t)$ (at an arbitrary moment of time $t\> 0$) and acting on linear operators $\eta$ on $\cH\ox \cF$ as
\begin{equation}
\label{Eee}
    \cE_{\lambda}(\eta)
    :=
    \re^{-\lambda \xi}\eta\re^{\lambda \xi}
    =
    \sum_{r\> 0}
    \frac{(-\lambda)^r}{r!}
    \ad_{\xi}^r(\eta)
    =
    \re^{-\lambda \ad_{\xi}}(\eta).
\end{equation}
Here, $\lambda$ is a real parameter, and  $\ad_{\xi}(\eta):= [\xi,\eta]=-\d_{\lambda}\cE_{\lambda}(\eta)\big|_{\lambda = 0}$ is the negative of the infinitesimal generator of the superoperators $\cE_{\lambda}$ with the group property $\cE_{\lambda}\circ \cE_{\mu} = \cE_{\lambda + \mu}$ and identity element $\cE_0$. A particular case of the group property, $\cE_{\lambda}\circ \cE_{-1/2} = \cE_{\lambda -1/2}$, implies that
\begin{align}
\label{swap1}
    \re^{\xi/2}\cE_{\lambda}(\eta)
    &= \cE_{\lambda-1/2}(\eta) \re^{\xi/2},\\
\label{swap2}
    \re^{\xi/2}\cE_{\lambda}(\eta)\cE_{\mu}(\zeta)
    &= \cE_{\lambda-1/2}(\eta) \cE_{\mu-1/2}(\zeta) \re^{\xi/2}.
\end{align}
The definition (\ref{Eee}) also shows that  the superoperator $\cE_{\lambda}$ carries out a similarity transformation and is, therefore,  homomorphic:
\begin{equation}
\label{homo}
    \cE_{\lambda}(\eta\zeta)
    =
    \cE_{\lambda}(\eta)
    \cE_{\lambda}(\zeta).
\end{equation}

\begin{lem}
\label{lem:EE}
The quantum Ito differential of the  exponential of the adapted quantum  process $\xi$ in (\ref{dxi}) is representable as
\begin{equation}
\label{dexi}
    \rd \re^{\xi}
     =
     \re^{\xi/2}
     (\alpha \rd t + \beta^{\rT} \rd W)
     \re^{\xi/2}.
\end{equation}
Here, 
\begin{align}
\label{alpha}
    \alpha
    & :=
    \int_{\!\!-1/2}^{1/2}
    \Big(
        \cE_{\lambda}(f)
        +
        \cE_{\lambda}(g)^{\rT} \Omega \int_{-1/2}^{\lambda} \cE_{\mu}(g)\rd \mu
    \Big)
    \rd \lambda,\\
\label{beta}
    \beta
    & :=
    (\beta_j)_{1\< j \< m}
    :=
    \int_{\!\!-1/2}^{1/2}
    \cE_{\lambda}(g)
    \rd \lambda
\end{align}
are adapted  quantum processes, self-adjoint at any time,  and the superoperator $\cE_{\lambda}$ from  (\ref{Eee}) applies to the vector $g$ entrywise.
\end{lem}
\begin{proof}
For any positive integer $r$,
repeated application of the  quantum Ito rule $\rd (\eta\zeta) = (\rd \eta) \zeta + \eta \rd \zeta + (\rd \eta) \rd \zeta $ to the $r$th power of $\xi$ yields
\begin{align*}
    \rd (\xi^r)
    = &
    \sum_{j=0}^{r-1}
    \xi^j (\rd \xi) \xi^{r-1-j}\\
    & +
    \sum_{j,k\>0:\, j+k \< r-2}
    \xi^j
    (\rd \xi)
    \xi^k
    (\rd \xi)
    \xi^{r-2-j-k},
\end{align*}
where the second sum represents the Ito correction term. Hence,
\begin{equation}
\label{sum12}
    \rd \re^{\xi}
    =
    \sum_{j,k\>0}
    \frac{\xi^j (\rd \xi) \xi^k}{(j+k+1)!}
     +
    \sum_{j,k,\ell\>0}
    \frac{\xi^j
    (\rd \xi)
    \xi^k
    (\rd \xi)
    \xi^{\ell}}{(j+k+\ell+2)!}.
\end{equation}
We will now use the following identities for Euler's multivariate  Beta functions
\begin{align}
\label{int1}
    \int_0^1
    (1-\lambda)^j
    \lambda^k
    \rd
    \lambda
    & =
    j!k!/(j+k+1)!,\\
\label{int2}
    \int_{\Lambda}
    (1-\lambda)^j
    (\lambda-\mu)^k\mu^{\ell}
    \rd \lambda
    \rd \mu
    & =
    j!k!\ell! /(j+k+\ell + 2)!,
\end{align}
which hold  for all nonnegative integers $j$, $k$, $\ell$, with the integration in (\ref{int2}) being carried out over a planar simplex
\begin{equation}
\label{simplex}
    \Lambda
    :=
    \{(\lambda,  \mu)\in \mR^2:\ 1\> \lambda\> \mu \> 0\}.
\end{equation}
In view of (\ref{int1}), the first sum in (\ref{sum12}) takes the form
\begin{align}
\nonumber
    \sum_{j,k\>0}&
    \frac{\xi^j (\rd \xi) \xi^k}{(j+k+1)!}
     =
    \sum_{j,k\>0}
    \int_{0}^1
    \frac{((1-\lambda)\xi)^j (\rd \xi) (\lambda \xi)^k}{j!k!}
    \rd \lambda\\
\label{sum1}
    &=
    \int_{0}^{1}
    \re^{(1-\lambda)\xi}
    (\rd \xi)
    \re^{\lambda \xi}
    \rd \lambda
    =
    \re^{\xi}
    \int_{0}^{1}
    \cE_{\lambda}
    (\rd \xi)
    \rd \lambda,
\end{align}
where (\ref{Eee}) is used. The right-hand side of (\ref{sum1}) is recognizable as the Gateaux derivative of the exponential  $\re^{\xi}$ \cite[Eq. (10)]{L_1968},  \cite[Eqs. (B5), (B6)]{S_1964}, \cite[Eqs. (2.1), (4.1)]{W_1967}, which,  in our context, is evaluated in the direction of the Ito differential $\rd \xi$.
Similarly, substitution of (\ref{int2}) into the second sum in (\ref{sum12}) yields
\begin{align}
\nonumber
    \sum_{j,k,\ell\>0}&
    \frac{\xi^j (\rd \xi) \xi^k (\rd \xi)\xi^{\ell}}{(j+k+\ell+2)!}\\
\nonumber
    =&
    \sum_{j,k,\ell\>0}
    \int_{\Lambda}
    \frac{((1-\lambda)\xi)^j (\rd \xi) ((\lambda-\mu) \xi)^k(\rd \xi)(\mu\xi)^{\ell}}{j!k!\ell !}
    \rd \lambda\rd \mu\\
\nonumber
    =&
    \int_{\Lambda}
    \re^{(1-\lambda)\xi}
    (\rd \xi)
    \re^{(\lambda-\mu) \xi}
    (\rd \xi)
    \re^{\mu \xi}
    \rd \lambda\rd \mu\\
\nonumber
    =&
    \re^{\xi}
    \int_{\Lambda}
    \cE_{\lambda}
    (\rd \xi)
    \cE_{\mu}
    (\rd \xi)
    \rd \lambda\rd \mu\\
\label{sum2}
    =&
    \re^{\xi}
    \int_{0}^1
    \cE_{\lambda}
    (\rd \xi)
    \Big(
    \int_{0}^{\lambda}
    \cE_{\mu}
    (\rd \xi)
    \rd \mu
    \Big)
    \rd \lambda,
\end{align}
where the multiple integral employs  the structure of the simplex (\ref{simplex}). This corresponds to the second  derivative of the exponential \cite[Eq. (11.6)]{W_1967}. By substituting (\ref{sum1}) and (\ref{sum2}) into (\ref{sum12}), it follows that
\begin{align}
\nonumber
    \rd \re^{\xi}
     =&
     \re^{\xi}
    \int_0^1
    \cE_{\lambda}(\rd \xi)
    \Big(
        1
        +
        \int_{0}^{\lambda}
        \cE_{\mu}(\rd \xi)
        \rd \mu
    \Big)
    \rd \lambda\\
\label{dexi0}
    =&
     \re^{\xi/2}
    \int_{\!\!-1/2}^{1/2}
    \cE_{\lambda}(\rd \xi)
    \Big(
        1
        +
        \int_{-1/2}^{\lambda}
        \cE_{\mu}(\rd \xi)
        \rd \mu
    \Big)
    \rd \lambda
     \re^{\xi/2}.\!\!
\end{align}
Here, the second equality is obtained by using the identities (\ref{swap1}), (\ref{swap2}) and appropriately translating the limits of integration.
Since the Ito differential $\rd W$ commutes with the adapted processes taken at the same or earlier moments of time, then
\begin{equation}
\label{Edxi}
    \cE_{\lambda}(\rd \xi) = \cE_{\lambda}(f) \rd t + \cE_{\lambda}(g)^{\rT}\rd W,
\end{equation}
where $\cE_{\lambda}(g)$ also commutes with $\rd W$.
By combining (\ref{Edxi}) with the quantum product rules $\rd t \rd W = 0$ and (\ref{Omega}), it follows that
\begin{equation}
\label{cEcE}
\cE_{\lambda}(\rd \xi) \cE_{\mu}(\rd \xi)
    =
    \cE_{\lambda}(g)^{\rT} \Omega\cE_{\mu}(g)    \rd t.
\end{equation}
The representation (\ref{dexi})  is now obtained by substituting (\ref{Edxi}), (\ref{cEcE}) into  (\ref{dexi0}) and using (\ref{alpha}), (\ref{beta}). Finally, the self-adjointness of the operators $\alpha(t)$ and $\beta_1(t), \ldots, \beta_m(t)$ can be  verified directly from (\ref{alpha}), (\ref{beta}) by using the symmetry of the integration limits $\pm 1/2$ about the origin and the identity
\begin{equation}
\label{cEsym}
    (\cE_{\lambda}(\eta))^{\dagger}
    =
    \cE_{-\lambda}(\eta^{\dagger})
\end{equation}
which follows from (\ref{Eee}) and $\xi^{\dagger} = \xi$. Indeed, (\ref{cEsym}) implies that $(\cE_{\lambda}(f))^{\dagger} = \cE_{-\lambda}(f)$ and  $(\cE_{\lambda}(g)^{\rT} \Omega \cE_{\mu}(g))^{\dagger} = \cE_{-\mu}(g)^{\rT} \Omega \cE_{-\lambda}(g)$, since the drift $f$ and  the entries of the dispersion vector $g$ in (\ref{dxi}) are self-adjoint operators and $\Omega = \Omega^*$. Hence,
$$
    \Big(
    \int_{\!\!-1/2}^{1/2}
    \cE_{\lambda}(f)\rd \lambda
    \Big)^{\dagger}
    =
    \int_{\!\!-1/2}^{1/2}
    \cE_{-\lambda}(f)\rd \lambda
    =
    \int_{\!\!-1/2}^{1/2}
    \cE_{\lambda}(f)\rd \lambda,
$$
and
\begin{align*}
    \Big(
    \int_{\Delta}
    \cE_{\lambda}(g)^{\rT} \Omega \cE_{\mu}(g)
    \rd \lambda \rd\mu
    \Big)^{\dagger}
    &=
    \int_{\Delta}
    \cE_{-\mu}(g)^{\rT} \Omega \cE_{-\lambda}(g)
    \rd \lambda \rd\mu    \\
    & =\int_{\Delta}
    \cE_{\lambda}(g)^{\rT} \Omega \cE_{\mu}(g)
    \rd \lambda \rd\mu,
\end{align*}
which proves that both parts of $\alpha$ in (\ref{alpha}) are self-adjoint, and a similar argument applies to the  entries of $\beta$ in (\ref{beta}).
Here, the set $\Delta:= \{(\lambda, \mu)\in \mR^2:\ 1/2\> \lambda \> \mu \> -1/2\}$, associated with (\ref{simplex}), is invariant under the area-preserving linear transformation
$(\lambda, \mu)\mapsto (-\mu,-\lambda)$. \end{proof}

If $\xi$ were a classical diffusion process, then $\rd \xi$ would commute  with $\xi$, thus making $\ad_{\xi}(\rd \xi)$ vanish and implying that $\cE_{\lambda}(\rd \xi) = \rd \xi$. In this commutative case, the representation (\ref{dexi0}) would reduce to $\rd \re^{\xi} = \re^{\xi}(\rd \xi  + (\rd \xi)^2/2)
$, leading to the Doleans-Dade exponential \cite{D_1970}, with the $1/2$-factor coming from  the area of the simplex (\ref{simplex}), and  $(\rd \xi)^2 = g^{\rT} \Omega g\rd t$ being the quadratic variation of $\xi$. Therefore, the relation (\ref{dexi}) can be regarded as a noncommutative quantum counterpart to the classical stochastic exponential.
The following theorem applies Lemma~\ref{lem:EE} to the dynamics of the exponential moment (\ref{Xi}).

\begin{thm}
\label{th:EEE}
The exponential moment (\ref{Xi}) of the quantum process $\xi$ from (\ref{dxi}) satisfies a differential equation
\begin{equation}
\label{Eexidot}
    \dot{\Xi}
     =
     \bE(\re^{\xi/2}\alpha\re^{\xi/2}),
\end{equation}
where the process $\alpha$ is defined by (\ref{alpha}).
\end{thm}
\begin{proof}
The commutativity between $\rd W$ and adapted processes, combined with the product structure of the system-field density operator (\ref{rho0}),    imply that
$$
    \bE(\re^{\xi/2}\beta^{\rT} \rd W\re^{\xi/2}) = \bE(\re^{\xi/2} \beta \re^{\xi/2})^{\rT} \bE \rd W = 0.
$$
Therefore, since the noise term on the right-hand side of (\ref{dexi}) does not contribute to  $\dot{\Xi}$,  the latter reduces to the average of the drift term $\re^{\xi/2} \alpha \re^{\xi/2}$, thus proving (\ref{Eexidot}).  \end{proof}

In the context of the exponential moment dynamics, the quantum process $\alpha$, which linearly enters   the right-hand side of (\ref{Eexidot}) and is computed according to (\ref{alpha}), will be referred to as the \emph{rate process}.

\section{Risk-sensitive moments of system variables}\label{sec:RSM}

As an adapted  quantum process $\xi$, we will now take a quadratic form of the system variables
\begin{equation}
\label{quadxi}
    \xi := X^{\rT} \Pi X/2,
\end{equation}
specified by a matrix $\Pi \in \mS_n$, so that $\xi(t)$ is a self-adjoint operator on $\cH\ox \cF$. Similarly to the usual quadratic forms of commuting variables, the diagonalization of $\Pi$ allows $\xi$ to be represented as a linear combination of squared observables
$$
    \xi
    =
    \frac{1}{2}
    \sum_{k= 1}^n
    \sigma_k (\nu_k^{\rT} X)^2
    \succcurlyeq
    \frac{\sigma_1 }{2}
    \sum_{k=1}^n
    X_k^2,
$$
whose coefficients $\sigma_1 \< \ldots \< \sigma_n$ are the eigenvalues of $\Pi$, with $\nu_1, \ldots, \nu_n$ the corresponding orthonormal eigenvectors in $\mR^n$.  Therefore, if $\Pi\succcurlyeq 0$, that is, $\sigma_1\> 0$, the operator $\xi(t)$ is also positive semi-definite. The exponential moment (\ref{Xi}), associated with (\ref{quadxi}), takes the form
\begin{equation}
\label{RSM}
    \Xi
    :=
    \bE \re^{X^{\rT} \Pi X/2}
\end{equation}
and will be referred to as the \emph{risk-sensitive moment} (RSM) of the system variables. This extends the risk-sensitive storage functions from classical stochastic systems \cite{MP_2009} to the quantum setting. In addition to the time dependence,  the RSM $\Xi$ depends on  $\Pi$, which plays the role of a matrix-valued risk-sensitivity  parameter. Its asymptotic behaviour for small $\Pi$ is described by
\begin{equation}
\label{Tay}
    \Xi
    =
    1
    +
    \bra
        \Pi,
        \Re \bE (XX^{\rT})
    \ket/2 + o(\Pi),
    \qquad
    \Pi \to 0,
\end{equation}
where $\bra K, L\ket:= \Tr (K^*L)$ denotes the Frobenius inner product of real or complex matrices,  and $\bE(XX^{\rT}) = \Re \bE (XX^{\rT}) + i\Theta/2$ is the matrix of second moments of the system variables.
Since $\re^{\xi} \succcurlyeq I + \xi$ (this property is inherited from the exponential function on the real line due to the spectral theorem for self-adjoint operators),  then
$    \Xi \>  \bE (I + \xi) = 1 + \bra \Pi, \Re \bE (XX^{\rT})\ket/2
$, so that the first two terms on the right-hand side of (\ref{Tay}) provide a lower bound for $\Xi$.
By a slightly refined reasoning, if $\Pi\succcurlyeq 0$,
then  $\re^{\xi}\succcurlyeq \xi^r/r!$, which yields an upper bound $\bE (\xi^r)\< r! \Xi$ for the higher-order moments of $\xi$ for all positive integers $r$. Now, to study the dynamics of the RSM $\Xi$ by the methods of Section~\ref{sec:exp_mom}, we will first describe the evolution of the process $\xi$.

\begin{lem}
\label{lem:dxi}
For the open quantum system, specified by Lemma~\ref{lem:dX}, the process $\xi$, defined by (\ref{quadxi}), satisfies the QSDE (\ref{dxi}) whose drift $f$ and dispersion vector $g$ are computed as
\begin{align}
\nonumber
    f
    =&
    \tau +
    (X^{\rT} (A^{\rT} \Pi + \Pi A)X\\
\label{f}
    &+ X^{\rT} \Pi \Theta C Z - Z^{\rT} C^{\rT} \Theta \Pi X
    )/2,\\
\label{tau}
        \tau
    :=& \bra BVB^{\rT}, \Pi\ket/2,\\
\label{g}
    g
    =&
    B^{\rT} \Pi X,
\end{align}
where the matrix $V$ is given by (\ref{V}).
\end{lem}
\begin{proof}
By  applying the quantum Ito formula to (\ref{quadxi}) and using (\ref{Omega}), (\ref{dX}), it follows that
\begin{align}
\nonumber
    2\rd& \xi
    =
    X^{\rT} \Pi \rd X + (\rd X)^{\rT} \Pi X + (\rd X)^{\rT} \Pi \rd X\\
\nonumber
    =&
    X^{\rT} \Pi (F\rd t + B\rd W) + (F\rd t + B\rd W)^{\rT} \Pi X
    + \rd W^{\rT}B^{\rT}\Pi B \rd W\\
\label{dquadxi}
    =&
    (X^{\rT} \Pi F + F^{\rT} \Pi X + \bra BVB^{\rT}, \Pi\ket)\rd t + 2X^{\rT} \Pi B \rd W,
\end{align}
where $\Tr(\Omega B^{\rT} \Pi B) = \bra BVB^{\rT}, \Pi\ket$ by the antisymmetry of the matrix $J$ in (\ref{CCRW}). It now remains to substitute (\ref{FAB}) into (\ref{dquadxi}) to verify that $\xi$ is governed by (\ref{dxi}) with (\ref{f})--(\ref{g}).
\end{proof}

Application of Theorem~\ref{th:EEE} to the RSM (\ref{RSM}) requires computation of the rate process $\alpha$, defined by (\ref{alpha}), for the drift $f$ and the dispersion vector $g$ of the QSDE (\ref{dxi}) which are specified by Lemma \ref{lem:dxi}.

\section{Computation of the rate process}\label{sec:RSM_diss}

The process $g$ in (\ref{g}) is linear with respect to $X$, whereas $f$ in (\ref{f}) is a quadratic function of  the system variables perturbed by the terms which are bilinear in $X$ and the Lur'e type nonlinearities $Z$ from (\ref{Z}). Since the processes $f$ and $g$ enter the rate process $\alpha$ in  (\ref{alpha}) as arguments of the superoperators $\cE_{\lambda}$, we need the following lemma which allows   $\cE_{\lambda}$ to be computed for linear, quadratic and more general nonlinear functions of the system variables by taking advantage of the CCRs (\ref{Theta}).

\begin{lem}
\label{lem:lin}
The action of the superoperator $\cE_{\lambda}$, associated with the system variables by (\ref{Eee}) and (\ref{quadxi}), on the vector $X$ is equivalent to the left multiplication by a matrix exponential:
\begin{equation}
\label{cEX}
    \cE_{\lambda}(X)
    =
    K_{\lambda} X,
    \qquad
    K_{\lambda}
    :=
    \re^{i\lambda \Theta \Pi}
    =
    \overline{K_{-\lambda}},
\end{equation}
where $\Theta$ is the CCR matrix of the system variables from (\ref{Theta}).
\end{lem}
\begin{proof}
The superoperator  $\ad_{\xi}$ applies to vectors entrywise   and commutes with the left multiplication by a complex matrix, so that $\ad_{\xi}(NX) = N \ad_{\xi}(X)$ for any  $N \in \mC^{n\x n}$. Hence, by using a straightforward induction, initialized with
\begin{equation}
\label{adxiX}
    \ad_{\xi}(X)
    =
    -i\Theta \Pi X
\end{equation}
in view of (\ref{Theta}) and (\ref{quadxi}), it follows that
$    \ad_{\xi}^r (X) = (-i\Theta \Pi)^r X$ for any nonnegative integer $r$. Substitution of the last relation into  (\ref{Eee}) yields
$
    \cE_{\lambda}(X)
    =
    \sum_{r\> 0}
    (i\lambda \Theta \Pi)^r
    X/r!
    =
    \re^{i\lambda \Theta \Pi} X
$,
which establishes (\ref{cEX}).
\end{proof}

Lemma~\ref{lem:lin} shows that the superoperator $\cE_{\lambda}$, associated  with (\ref{quadxi}), performs a Lie algebraic similarity transformation \cite[Section 6]{W_1967} of the system variables. In particular, this implies that
\begin{equation}
\label{cEpsi}
    \cE_{\lambda}(\psi(u^{\rT} X))
    =
    \psi(u^{\rT} K_{\lambda} X)
\end{equation}
for any vector $u \in \mC^n$ and at least for analytic functions $\psi$ of a complex variable. It suffices to verify (\ref{cEpsi}) for monomials $\psi(z):= z^r$ of arbitrary nonnegative integer degrees $r$. To this end, note that
$$
    \cE_{\lambda}((u^{\rT}X)^r)
    =
    (\cE_{\lambda}(u^{\rT}X))^r
    =
    (u^{\rT}\cE_{\lambda}(X))^r
    =
    (u^{\rT}K_{\lambda}X)^r,
$$
which is obtained by combining the homomorphic property (\ref{homo}) with Lemma~\ref{lem:lin}.
Another identity, which follows from (\ref{adxiX}), is given by
\begin{align}
\nonumber
    \ad_{\xi}(XX^{\rT})
    & =
    \ad_{\xi}(X) X^{\rT} + X \ad_{\xi}(X)^{\rT}\\
\label{adxiXX}
    & =
    i(XX^{\rT}\Pi\Theta - \Theta \Pi XX^{\rT}),
\end{align}
where use is also made of the Leibniz product rule $[\xi,\eta\zeta] = [\xi, \eta]\zeta + \eta [\xi, \zeta]$ for the commutator.
The fact that the right-hand side of (\ref{adxiXX}) is quadratic with respect to $X$ (as is $\xi$ in (\ref{quadxi})),
is closely related to the property  that the commutator of quadratic forms of annihilation and creation operators is also a quadratic form of these operators \cite[Lemma 4]{PUJ_2012}; see also \cite[Appendix B]{JK_1998} for other relevant results. Now, the linearity of $g$ in (\ref{g}) with respect to the system variables allows Lemma~\ref{lem:lin} to be applied as
\begin{equation}
\label{cEg}
    \cE_{\lambda}(g)
    =
    B^{\rT} \Pi
    \cE_{\lambda}(X)
    =
    B^{\rT} \Pi K_{\lambda} X.
\end{equation}
For what follows, we assume that the system dimension $n$ is even,  and
\begin{equation}
\label{good}
    \det \Theta \ne 0,
    \qquad
    \Pi\succ 0.
\end{equation}
Then, in view of
$\d_{\lambda}K_{\lambda} = i\Theta \Pi K_{\lambda}$,
the representation (\ref{cEg}) implies that
\begin{align}
\nonumber
    \int_{-1/2}^{\lambda}
    \cE_{\mu}(g)\rd \mu
    &=
    B^{\rT} \Pi
    \int_{-1/2}^{\lambda}
    K_{\mu}\rd \mu
    X\\
\nonumber
    & =
    B^{\rT} \Pi (i\Theta \Pi)^{-1}
    (K_{\lambda} - K_{-1/2})X\\
\label{intcEmu}
    & =
    iM
    (K_{\lambda} - K_{-1/2})X,
\end{align}
where the matrix $M=-B^{\rT}\Theta^{-1}$ specifies the system-field coupling (\ref{hM}).
It will be convenient to extend the notation $\bra K,L\ket $ to the case of a real or  complex  matrix
$K:=(K_{jk})$ and an appropriately dimensioned  matrix $L:= (L_{jk})$ of operators on a common space by $\bra K,L\ket:=  \sum_{j,k} \overline{K_{jk}}L_{jk}$.  Substitution of (\ref{cEg}) and (\ref{intcEmu}) (which both are linear with respect to $X$) into the Ito correction part of the rate process $\alpha$ in (\ref{alpha}) leads to a quadratic form of the system variables
\begin{equation}
\label{Egg}
    \!\int_{\!\!-1/2}^{1/2}
    \Big(
    \cE_{\lambda}(g)^{\rT}
    \Omega
    \int_{-1/2}^{\lambda}
    \cE_{\mu}(g)
    \rd \mu
    \Big)
    \rd \lambda
    \!=\!
    X^{\rT} \Gamma X
    \!=\!
    \bra
        \overline{\Gamma},
        XX^{\rT}
    \ket,\!\!\!\!\!
\end{equation}
where $\Gamma$ is a complex Hermitian matrix of order $n$, which is computed as
\begin{equation}
\label{Gamma}
 \Gamma
 :=i
    \int_{\!\!-1/2}^{1/2}
    K_{\lambda}^{\rT}
    \Pi B
    \Omega
    M(K_{\lambda}-K_{-1/2})
    \rd \lambda.
\end{equation}
In view of (\ref{good}), the matrix  $\Theta \Pi = \Pi^{-1/2}\sqrt{\Pi}\Theta \sqrt{\Pi} \sqrt{\Pi}$ is obtained by a similarity transformation from $\sqrt{\Pi}\Theta \sqrt{\Pi} \in \mA_n$, with
$\sqrt{\Pi}$ a real positive definite symmetric matrix square root of $\Pi$.
Hence, $\Theta \Pi$ is a diagonalizable matrix  whose eigenvalues are pure imaginary and symmetric about the origin  \cite{HJ_2007}, that is, representable as $\pm i \omega_k$,  where $\omega_1, \ldots, \omega_{n/2}$ are all real. In this case, the matrix exponential $K_{\lambda}$ in (\ref{cEX}) is isospectral to a complex positive definite Hermitian matrix  $\re^{i\lambda \sqrt{\Pi}\Theta \sqrt{\Pi}}$
with eigenvalues $\re^{\pm\lambda \omega_k}> 0$ (recall that the parameter $\lambda$ is real). Thus,
\begin{equation}
\label{Keig}
    K_{\lambda} = \Psi D_{\lambda}  \Psi^{-1},
\end{equation}
where the columns of the matrix $\Psi \in \mC^{n\x n}$ are the eigenvectors of $\Theta \Pi$, and
\begin{equation}
\label{D}
    D_{\lambda}:= \diag_{1\< j \< n}(d_j(\lambda)) = \diag_{1\< k \< n/2}(\re^{\pm\lambda \omega_k})
\end{equation}
is a diagonal matrix formed by the eigenvalues of $K_{\lambda}$. We will now consider a linear operator $\cK$, which acts on an $(n\x n)$-matrix $P$ as
\begin{equation}
\label{cK}
    \cK(P)
    :=
    \int_{\!\!-1/2}^{1/2}
    K_{\lambda} P K_{\lambda}^{\rT} \rd \lambda
    =
    \Psi(S\od (\Psi^{-1} P\Psi^{-\rT}))\Psi^{\rT},
\end{equation}
where use is made of (\ref{Keig}). Here, $\od$ denotes the Hadamard product of matrices, and  $\Psi^{-\rT} := (\Psi^{-1})^{\rT}$. Also, the entries of the matrix $S:= (s_{jk})_{1\< j,k\< n}  \in \mS_n$ are computed in terms of (\ref{D}) as
\begin{equation}
\label{S}
    s_{jk} :=
    \int_{\!\!-1/2}^{1/2}
    d_j(\lambda)d_k(\lambda)\rd \lambda.
\end{equation}
and are all strictly positive. Therefore,  the operator $\cK$ in (\ref{cK}) is invertible, and its inverse has a similar representation
\begin{equation}
\label{cKinv}
    \cK^{-1}(P)
    =
    \Psi(S^{\om}\od (\Psi^{-1} P\Psi^{-\rT}))\Psi^{\rT},
\end{equation}
where $S^{\om}:= (1/s_{jk})_{1\< j,k\< n}$ denotes the entrywise inverse of the matrix $S$. Also, the subspace of complex Hermitian matrices of order $n$ is invariant under $\cK$. The significance of the operator $\cK$ is clarified by the relation
\begin{align}
\nonumber
    \int_{\!\!-1/2}^{1/2}
    \cE_{\lambda}(XX^{\rT})\rd \lambda
    & =
    \int_{\!\!-1/2}^{1/2}
    \cE_{\lambda}(X)\cE_{\lambda}(X)^{\rT}\rd \lambda\\
\label{cKXX}
    & =
    \int_{\!\!-1/2}^{1/2}
    K_{\lambda}X X^{\rT}K_{\lambda}^{\rT}\rd \lambda
    =
    \cK(XX^{\rT}),
\end{align}
where use is made of the homomorphic property (\ref{homo}) and Lemma~\ref{lem:lin}. In the same vein, we define  an $(s\x n)$-matrix of operators
\begin{equation}
\label{Q}
    Q
    :=
    \int_{\!\!-1/2}^{1/2} \cE_{\lambda}(ZX^{\rT})\rd \lambda
    =
    \int_{\!\!-1/2}^{1/2} \cE_{\lambda}(Z)X^{\rT} K_{\lambda}^{\rT}\rd \lambda,
\end{equation}
where the rightmost equality employs the corollary (\ref{cEpsi}) of Lemma~\ref{lem:lin} applied to (\ref{Z}) as
\begin{equation}
\label{cEZ}
    \cE_{\lambda}(Z)
    =
    (\varphi_k' (c_k^{\rT} K_{\lambda}X))_{1\< k \< s}.
\end{equation}
Now, by substituting (\ref{f}), (\ref{Egg}) into (\ref{alpha}), and using (\ref{cKXX}) and (\ref{Q}), it follows that the rate process for the RSM (\ref{RSM})  takes the form
\begin{align}
\nonumber
    \alpha
    = &
    \tau +
    \bra
        A^{\rT} \Pi + \Pi A,
        \cK(XX^{\rT})
    \ket/2
     +
    \bra
        \overline{\Gamma}, XX^{\rT}
    \ket +\bra \Pi, T\ket/2\\
\nonumber
    =&
    \tau +
    \bra
        A^{\rT} \Pi + \Pi A + 2 \cK^{-\dagger}(\overline{\Gamma}),
        \cK(XX^{\rT})
    \ket/2\\
\label{RSMalpha}
    & + \bra \Pi, T\ket/2.
%
\end{align}
Here, $T$ is an $(n\x n)$-matrix of operators associated with (\ref{Q}) by
\begin{equation}
\label{T}
    T :=
    \Theta C Q - Q^{\dagger} C^{\rT} \Theta.
\end{equation}
Also, $\cK^{-\dagger}:= (\cK^{-1})^{\dagger}$ denotes the adjoint of the inverse operator $\cK^{-1}$ from (\ref{cKinv}) in the Hilbert space $\mC^{n\x n}$ with the Frobenius inner product of matrices, which is computed as
\begin{equation}
\label{cKinvadj}
    \cK^{-\dagger}(P)
    =
    \Psi^{-*}(S^{\om}\od (\Psi^* P\overline{\Psi}))\overline{\Psi^{-1}}.
\end{equation}

\section{Dynamics of risk-sensitive moments}\label{sec:RSM_dyn}

Substitution of the rate process (\ref{RSMalpha}) into (\ref{Eexidot}) of Theorem~\ref{th:EEE} yields the following equation for the RSM of the system
\begin{align}
\nonumber
    \d_t \Xi
    =&
    \tau \Xi +
    \bra
        A^{\rT} \Pi + \Pi A + 2 \cK^{-\dagger}(\overline{\Gamma}),
        N)
    \ket/2\\
\label{Xidot}
    & +
    \Bra
        \Pi,
        \bE(\re^{\xi/2}T\re^{\xi/2})
    \Ket/2.
\end{align}
Here, $N$ is a complex Hermitian matrix of order $n$ defined by
%
%
\begin{align}
\nonumber
    N
    := &
    \bE
    \Big(
        \re^{\xi/2}
        \int_{\!\!-1/2}^{1/2}
        \cE_{\lambda}(XX^{\rT})
        \rd \lambda
        \re^{\xi/2}
    \Big)\\
\label{N}
    =&
    \cK(\bE(\re^{\xi/2}XX^{\rT}\re^{\xi/2})).
\end{align}
The following lemma directly relates the matrix $N$ with the RSM.

\begin{lem}
\label{lem:N}
The matrix $N$, defined by (\ref{N}), is representable in terms of the RSM $\Xi$ from (\ref{RSM}) by a linear differential operator
\begin{equation}
\label{Nfull}
    N = 2\d_{\Pi} \Xi + i\Xi\Theta /2,
\end{equation}
where $\Theta$ is the CCR matrix of the system variables from (\ref{Theta}).
\end{lem}
\begin{proof}
Upon splitting $XX^{\rT}$ into the symmetric and antisymmetric parts as
\begin{equation}
\label{XX}
    XX^{\rT} = (XX^{\rT} + (XX^{\rT})^{\rT})/2 + i\Theta/2,
\end{equation}
the imaginary part of the matrix $N$ in (\ref{N}) can be expressed in terms of (\ref{RSM}) as
\begin{equation}
\label{ImN}
    \Im N
    =
    \bE
    \Big(
        \re^{\xi/2}
        \int_{\!\!-1/2}^{1/2}
        \cE_{\lambda}(\Theta/2)
        \rd \lambda
        \re^{\xi/2}
    \Big)
    =
    \Xi \Theta /2,
\end{equation}
where use is made of the identity $\cE_{\lambda}(\Theta) = \Theta$ (recall that constants remain unchanged under the action of $\cE_{\lambda}$). Now, the process $\xi$ in (\ref{quadxi}) depends linearly on the matrix $\Pi\in \mS_n$, and
$$
    \d_{\Pi} \xi
    =
    (XX^{\rT} + (XX^{\rT})^{\rT})/4
$$
coincides with the symmetric part of $XX^{\rT}$ in (\ref{XX}) up to a factor of 2.
Hence,  a reasoning, similar to the proofs of Lemma~\ref{lem:EE} and Theorem~\ref{th:EEE} (except that no Ito correction terms arise in this case),  allows the real part of $N$ in (\ref{N}) to be related to the Frechet derivative of the RSM $\Xi$ with respect to  $\Pi$ as
\begin{align}
\nonumber
    \Re N
    & =
    \bE
    \Big(
        \re^{\xi/2}
        \int_{\!\!-1/2}^{1/2}
        \cE_{\lambda}((XX^{\rT} + (XX^{\rT})^{\rT})/2)
        \rd \lambda
        \re^{\xi/2}
    \Big)\\
\label{ReN}
    & =
    2
    \bE
    \Big(
        \re^{\xi/2}
        \int_{\!\!-1/2}^{1/2}
        \cE_{\lambda}(\d_{\Pi}\xi)
        \rd \lambda
        \re^{\xi/2}
    \Big)
    =
    2\d_{\Pi} \Xi.
\end{align}
The representation (\ref{Nfull}) now follows from (\ref{ImN}) and (\ref{ReN}).
\end{proof}

Lemma~\ref{lem:N} allows the differential equation (\ref{Xidot}) to be written in a more closed form.

\begin{thm}
\label{th:RSM}
Under the conditions (\ref{good}), the RSM (\ref{RSM})  of the quantum system, described by Lemma~\ref{lem:dX}, satisfies a partial differential equation (PDE)
\begin{align}
\nonumber
    \d_t \Xi
    =&
    \bra
        A^{\rT}\Pi + \Pi A + 2\Ups, \d_{\Pi}\Xi
    \ket
    +
    (\tau + \bra \mho, \Theta\ket/2) \Xi\\
\label{Xidot1}
    & +
    \bra
        \Pi, \bE(\re^{\xi/2} T \re^{\xi/2})
    \ket/2.
\end{align}
Here,
\begin{equation}
\label{Upsmho}
    \Ups := \Re \cK^{-\dagger}(\overline{\Gamma}),
    \qquad
    \mho := \Im \cK^{-\dagger}(\overline{\Gamma})
\end{equation}
are computed using (\ref{Gamma}), (\ref{cKinvadj}), and $\tau$ and $T$ are defined by (\ref{tau}), (\ref{T}).
\end{thm}

Theorem~\ref{th:RSM}, which is the main result of the paper, shows that the unperturbed part of the RSM dynamics (the first line of (\ref{Xidot1}))  corresponds to a linear PDE of first order which can, in principle,  be solved by the method of characteristics \cite{E_1998} using the solutions of the ordinary differential equations (ODEs)
\begin{equation}
\label{Pidot}
    \dot{\Pi}
    +
    A^{\rT}\Pi + \Pi A + 2\Ups=0.
\end{equation}
This requires computation of the matrix-valued functions $\Ups$, $\mho$  of the risk-sensitivity parameter $\Pi\succ 0$ in (\ref{Upsmho}) and will be discussed elsewhere. The influence of the Lur'e type perturbations in the system Hamiltonian is described by the second line of (\ref{Xidot1}) and will be considered in the next section.


\section{Perturbations bounded using superpositive ordering}\label{sec:super}

Let $L := (L_{jk})_{1\< j,k\< r}$ be a matrix whose entries are linear operators on the system-field Hilbert space $\cH \ox \cF$, satisfying $L^{\dagger} := (L_{kj}^{\dagger})_{1\< j,k\< r} = L$, where, in application to matrices of operators, $(\cdot)^{\dagger} := ((\cdot)^{\#})^{\rT}$ denotes the transpose of the entrywise adjoint  $(\cdot)^{\#}$. Then for any $r$-dimensional complex vector $u:= (u_j)_{1\< j\< r}$, the operator $u^* L u := \sum_{j,k=1}^r \overline{u_j} L_{jk} u_k$ is self-adjoint. We say that such a matrix $L$  is \emph{superpositive} if $u^*L u\succcurlyeq 0$ for any $u \in \mC^r$. Since the superpositiveness extends the standard positive semi-definiteness (from a single operator to a matrix of operators), it will be written using the same symbol as $L \succcurlyeq 0$. The superpositiveness of $L$ is inherited by the matrix $\eta L \eta^{\dagger}$ for any linear operator $\eta$ on $\cH\ox \cF$. Also, $L \succcurlyeq  0$ implies that $\bE L\succcurlyeq 0$ holds regardless of the reference density operator over which the expectation is computed. Similarly to the usual positive semi-definiteness, the superpositiveness induces a partial  ordering among matrices with operator-valued entries. A straightforward example of a superpositive matrix is $XX^{\rT}$. We will now use the concept of superpositiveness to specify a class of Lur'e type perturbations as those which satisfy
\begin{equation}
\label{class}
    T
    \preccurlyeq
    \sigma
    \cK(XX^{\rT}),
\end{equation}
where $\sigma>0$ is a scalar parameter, and use is made of (\ref{cKXX}), (\ref{T}). The matrix $T$ in (\ref{T})
depends on the risk-sensitivity parameter $\Pi\succ 0$ through the matrix $Q$ in (\ref{Q}) since the matrix exponential $K_{\lambda}$ in (\ref{cEX}) involves $\Pi$.  The effect of this dependence on the description (\ref{class}) of the class of perturbations requires additional analysis  which will be carried out elsewhere. To this end, we will only mention here that in a particular case of quadratic perturbations $\varphi_1, \ldots, \varphi_s$ given by (\ref{quadphi}), the matrix $Q$ becomes linearly related to the matrix on the right-hand of (\ref{class}) as
\begin{equation}
\label{Qspec}
    Q = \diag_{1\< k\< s}(\gamma_k) C^{\rT} \cK(XX^{\rT}).
\end{equation}
Substitution of (\ref{Qspec}) into (\ref{T}) shows that the corresponding  class (\ref{class})  of (quadratic) perturbations of the Hamiltonian can be described in terms of a Lyapunov operator $P\mapsto LP+ P L^{\rT}$, where the matrix $L:= \Theta C\diag_{1\< k\< s}(\gamma_k) C^{\rT}$ does not depend on $\Pi$. We will now discuss the role of the inequality (\ref{class}), which employs the concept of superpositiveness, in the dynamics of the RSM (\ref{RSM}) for general (not necessarily quadratic)  Lur'e type perturbations.

\begin{thm}
\label{th:T}
Under the conditions of Theorem~\ref{th:RSM}, and for the Lur'e type perturbations described by (\ref{class}), the RSM of the quantum system satisfies a partial differential inequality (PDI)
\begin{align}
\nonumber
    \d_t \Xi
    \< &
    \bra
        A^{\rT}\Pi + \Pi A + 2\Ups + \sigma \Pi, \d_{\Pi}\Xi
    \ket\\
\label{Xidot2}
    & +
    (\tau + \bra \mho, \Theta\ket/2) \Xi.
\end{align}
\end{thm}
\begin{proof}
From (\ref{class}) and the properties of the superpositiveness,  it follows that
\begin{align}
\nonumber
    \bE(\re^{\xi/2} T \re^{\xi/2})
    &\preccurlyeq
    \sigma
    \bE(\re^{\xi/2} \cK(XX^{\rT}) \re^{\xi/2})\\
\label{pos1}
    &
    =\sigma N
    =
    \sigma(2\d_{\Pi} \Xi + i\Xi\Theta /2),
\end{align}
where use is made of Lemma~\ref{lem:N}. Since $\Pi\succ 0$, then (\ref{pos1}) leads to an upper bound for the perturbation term on the right-hand side of (\ref{Xidot1}):
\begin{align}
\nonumber
    \bra
        \Pi, \bE(\re^{\xi/2} T \re^{\xi/2})
    \ket/2
    & \<
    \sigma
    \bra
        \Pi,
        2\d_{\Pi} \Xi + i\Xi\Theta /2
    \ket/2\\
\label{pos2}
    & =
    \sigma
    \bra
        \Pi,
        \d_{\Pi} \Xi
    \ket,
\end{align}
where $\bra \Pi, \Theta\ket = 0$ since $\Pi$, $\Theta$ belong to the orthogonal subspaces $\mS_n$, $\mA_n$. The PDI (\ref{Xidot2}) is now obtained by combining (\ref{Xidot1}) with (\ref{pos2}).
\end{proof}

The PDI (\ref{Xidot2}) can be treated as a multivariate partial differential case of the Gronwall-Bellman lemma.  In fact, it is reduced to the standard univariate version of the lemma by noting that
$
    \d_t \Xi
    -
    \bra
        A^{\rT}\Pi + \Pi A + 2\Ups + \sigma \Pi, \d_{\Pi}\Xi
    \ket
$ coincides with the total time derivative $\rd \Xi /\rd t$ of the RSM $\Xi$ (as a function of $t$ and $\Pi$) along the trajectories of another characteristic  ODE
\begin{equation}
\label{Pidot1}
    \dot{\Pi}
    +
    A^{\rT}\Pi + \Pi A + 2\Ups +\sigma \Pi=0,
\end{equation}
which  is similar to (\ref{Pidot}). Indeed, along the characteristics (\ref{Pidot1}), the PDI (\ref{Xidot2}) is equivalent to an ordinary differential inequality
$$
    \rd \Xi /\rd t \< (\tau + \bra \mho, \Theta\ket/2) \Xi,
$$
thus making the univariate Gronwall-Bellman lemma applicable to this case.

\end{document}